\documentclass[reqno]{amsart}
\usepackage{amsmath,amsfonts,amsgen,amstext,amsbsy,amsopn,amsthm, amssymb}

\date{June 2, 2011}

\newtheorem{Theorem}{Theorem}
\newtheorem{proposition}{Proposition}
\newtheorem{Lemma}{Lemma}

\theoremstyle{remark}
\newtheorem{remark}{Remark}

\numberwithin{equation}{section}

\newcommand{\R}{\mathbb{R}}

\newcommand{\F}{\mathcal{F}}

\newcommand{\be}{\begin{equation}}
\newcommand{\ee}{\end{equation}}
\newcommand{\bea}{\begin{align}}
\newcommand{\eea}{\end{align}}

\newcommand\infspec{{\rm{inf\, spec\,}}}

\newcommand\bfu{\mathbf{u}}

\begin{document}

\title{Binding of Polarons and Atoms at Threshold}

\author[R.L. Frank]{Rupert L. Frank} \address{R.L. Frank, Department of Mathematics, Princeton University, Washington Road, Princeton, NJ 08544, USA} \email{rlfrank@math.princeton.edu}

\author[E.H. Lieb]{Elliott H. Lieb}
 \address{E.H. Lieb, Departments of Mathematics and Physics, Princeton University, P.O. Box 708, Princeton, NJ 08544, USA} \email{lieb@princeton.edu}

\author[R. Seiringer]{Robert Seiringer}
 \address{R. Seiringer, Department of Mathematics and Statistics, McGill
  University, 805 Sherbrooke Street West, Montreal, QC H3A 2K6,
  Canada} 
\email{rseiring@math.mcgill.ca}

\begin{abstract}
If the polaron coupling constant $\alpha$ is large enough, bipolarons
or multi-polarons will form. When passing through the critical
$\alpha_c$ from above, does the radius of the system simply get arbitrarily
large or does it reach a maximum and then explodes? We prove that it
is always the latter. We also prove the analogous statement for the
Pekar-Tomasevich (PT) approximation to the energy, in which case there is a
solution to the PT equation at $\alpha_c$. Similarly, we show that the
same phenomenon occurs for atoms,
e.g., helium, at the critical value of the nuclear charge. 
Our proofs rely only on energy estimates, not on a detailed
analysis of the Schr\"odinger equation, and are very general. They
use the fact that the Coulomb repulsion decays like
$1/r$, while `uncertainty principle' localization energies decay more
rapidly, as $1/r^2$.
\end{abstract}

\maketitle
\renewcommand{\thefootnote}{${}$} \footnotetext{\copyright\, 2011 by
  the authors. This paper may be reproduced, in its entirety, for
  non-commercial purposes.}


\section{Introduction}

We investigate the binding-unbinding transition of the ground state of
multi-polaron systems, both in the original Fr\"ohlich model \cite{Fr} and in
the Pekar-Tomasevich large coupling approximation \cite{PeTo}. The parameter
that varies is the
strength of the electron-electron Coulomb repulsion, conventionally denoted by
$U>0$. (We could, alternatively, let the field coupling constant $\alpha$ vary,
but it is mathematically convenient and equivalent to vary $U$.)

The conclusion of our rigorous analysis is that this transition is always
`first-order' in accordance with a proposal of Verbist, Peeters and Devreese
\cite{VPD}. That is, the Coulomb repulsion jumps discontinuously from a positive
value to zero and the radius, too, jumps discontinuously.

The same thing happens for the helium atom unbinding as $U$ varies, as shown by
M. and T. Hoffmann-Ostenhof and Simon \cite{HoHoSi}. Our analysis does not
proceed, as in \cite{HoHoSi}, by analyzing the Schr\"odinger equation, but,
instead, is based on energy estimates. Indeed, our proof is also valid in the
presence of a magnetic field; the proof in \cite{HoHoSi} cannot accommodate a
magnetic field because it relies on the positivity of the ground state wave
function. To illustrate our approach we first show how to reproduce the result
of \cite{HoHoSi} as a warm-up to the harder polaron problem. The overall lesson
is that \emph{this kind of discontinuous binding will occur whenever the net
repulsion at large distances falls of slower than $r^{-2}$.} This conclusion
does not depend on Bose or Fermi statistics.

Fr\"ohlich's Hamiltonian for $N$ particles interacting with the longitudinal
distortions of a polar crystal is
\begin{equation}
H^{(N)}_U = \sum_{i=1}^N \left(p_i^2 - \sqrt\alpha\phi(x_i) \right) +
H_f + U\, V_C(X)
\end{equation}
with $X=(x_1,\dots,x_N)\in \R^{3N}$ and 
\begin{equation}
V_C(X) = \sum_{i<j} \frac{1}{|x_i-x_j|} \,.
\end{equation}
The Hilbert space is $L^2(\R^{3N})\otimes \F$, where $\mathcal F$ is the bosonic
Fock space for the longitudinal optical modes of the crystal, with scalar
creation and annihilation operators $a^\dagger(k)$ and $a(k)$ satisfying
$[a(k),a^\dagger(k')]=\delta(k-k')$. The phonon field energy is
\begin{equation}\label{Hf.eq}
H_f = \int_{\R^3} dk\, a^\dagger(k) a(k)
\end{equation}
and the interaction of the crystal modes with the electron is
\begin{equation}
\phi(x) = \frac{1}{\sqrt2\pi} \int_{\R^3} \frac{dk}{|k|} \left( e^{ikx} a(k)
  + e^{-ikx} a^\dagger(k) \right) \,,
\end{equation}
with coupling constant $\alpha>0$. We define the \emph{ground state energy} to
be
\begin{equation}
 \label{eq:gse}
E^{(N)}_U(\alpha) = \infspec H^{(N)}_U = \inf_{\|\Psi\|=1} \langle \Psi | 
H^{(N)}_U |\Psi\rangle
\end{equation}
and the \emph{break-up energy} to be
\begin{equation}
 \label{eq:mbe}
\widetilde E^{(N)}_U(\alpha) = \min_{1\leq n\leq N-1}
\left( E^{(n)}_U(\alpha) + E^{(N-n)}_U(\alpha) \right)\,.
\end{equation}
It is always the case that $E^{(N)}_U(\alpha)\leq \widetilde
E^{(N)}_U(\alpha)$. If there is strict inequality, we say that the
$N$-particle system is \emph{bound}.

\emph{A note on particle statistics.} As we said, the results in this paper do
not depend
on Bose statistics (restrictions to symmetric spatial functions) or Fermi
statistics (restrictions to anti-symmetric functions of space and spin). The
numerical values of the binding energies and the discontinuities, etc., will
depend on statistics, but not the qualitative features
proved here. 

Let us start by recalling some previous results. Some of
these results are only
valid in the Pekar-Tomasevich (PT) approximation, which will be described later.

\begin{enumerate}
 \item Physically, $U>2\alpha$, but the case $U\leq 2\alpha$ has been
considered in \cite{GrMo}. They show that for \emph{fermions} the energy
satisfies $-C_1 N^{7/3} \leq E^{(N)}_U(\alpha) \leq -C_2 N^{7/3}$. The same
arguments for bosons lead to $-C_1' N^{3} \leq E^{(N)}_U(\alpha) \leq -C_2'
N^{3}$. In \cite{BeBl,BeFrLi} the sharp asymptotic constants for bosons and
fermions are identified in the PT approximation.

\item The case $U=2\alpha$ is special. In the context of the PT model
\cite{GrMo} showed that the energy behaves linearly in $N$ for fermions and in
\cite{BeFrLi} an $N^{7/5}$ behavior for bosons is proved.

\item For $U>2\alpha$ the thermodynamic limit, $\lim_{N\to\infty} E^{(N)}_U(\alpha)/N$, exists
\cite{FLST}. More importantly, it was shown in \cite{FLST} that there is a value $U^*(\alpha)$
such that $E^{(N)}_U(\alpha)=N E^{(1)}(\alpha)$ for all $U\geq U^*(\alpha)$
and all $N$. (For $N=1$, we drop the index $U$ in the notation
$E^{(N)}_U(\alpha)$.) In other words, there is no binding of any kind for $U>
U^*(\alpha)$.

\item There is, in fact, binding for some range of $\alpha$, $U$ and $N$ with
$U>2\alpha$. (If $N=2$ and $\alpha$ is large then binding occurs for $2\alpha<U
\lesssim 2.3\, \alpha$ \cite{SuMo,fomin,VSPD}.)
\end{enumerate}

For every fixed $N$ and $\alpha$ there will generally be certain critical
values of $U$, which we generically denote by $U_c$, at which binding will
either appear or disappear. That is,
$$
E^{(N)}_U(\alpha)=\widetilde E^{(N)}_{U_c}(\alpha)
$$
and
$$
E^{(N)}_U(\alpha) < \widetilde E^{(N)}_U(\alpha)
$$
for either all $U\in (U_c-\delta,U_c)$ or all $U\in (U_c,U_c+\delta)$,
for some $\delta>0$. The usual case is that there is a unique $U_c$
such that binding occurs for $U<U_c$ and no binding occurs for
$U>U_c$, but we do not know if this is true and we certainly do not
know how to prove it. The problem is that the energies of the break-up
components also depend on $U$, so there is no obvious monotonicity.

The question we address here is what happens at $U_c$, or as $U$ approaches
$U_c$ from the binding side. Does the compound simply explode into two pieces
like a supernova or does it get larger and larger like a red giant? Our
answer is that it always explodes, that is, the maximal distance between the
particles stays finite. For the bipolaron this was seen variationally in
\cite{VPD} and stated as a general mathematical question in \cite{spohn}.

One would like to say that this implies the existence of an eigenfunction when
$U=U_c$, but the polaron problem presumably never has bound states because of
translation invariance \cite{GeLo1}. (No doubt there are bound states of total
momentum equal to zero, but we will not explore this here.) The PT problem, on
the other hand, does have honest energy minimizers in a one-sided neighborhood
of $U_c$ \cite{Li,lewin}, and we use the result about the finite maximal
distance to show, in a separate theorem, that the PT problem has a minimizer
when $U=U_c$.

In any case, a first order phase transition occurs. 
Suppose there is
a break-up into two complexes of $n$ and $N-n$ particles each. Suppose particle
one goes into the first complex and particle two goes into the second. Then
their mutual Coulomb repulsion is zero after the break-up, but it is uniformly
(in $U$) bounded away from zero before the break-up.

While in a physical situation one is usually not at the critical value of the
parameters, the results proved here can be useful for a general understanding
of quantum mechanics and, perhaps, for numerical analysis close to the critical
values.


\section{The helium problem}

As a warm-up, we consider the ground state of a two-electron atom, of which
helium and the hydrogen ion are examples. The Hamiltonian
$$
H_U = \sum_{i=1}^2 \left(p_i^2 - |x_i|^{-1} \right) + \frac U{|x_1-x_2|}
$$
acts on $L^2(\R^6)$. The ground state energy $E_U$ of $H_U$ is
monotone increasing with respect to $U$, and there is a critical $U_c>0$ such
that for $U<U_c$, $E_U$ is lower than
$$
-\frac14 = \infspec \left( p^2 - |x|^{-1} \right) \,,
$$
which is the energy of a hydrogen atom, and for $U\geq U_c$ they are equal. Our
general method, which will later be applied to the polaron problem, is
illustrated by the following theorem of M. and T. Hoffmann-Ostenhof and Simon
\cite{HoHoSi}.

\begin{Theorem}[\textbf{Binding for helium
\cite{HoHoSi}}]\label{hohosi}
 $H_{U_c}$ has a ground state eigenfunction $\psi_{U_c}$ in $L^2(\R^6)$, i.e.,
$H_{U_c}\psi_{U_c} = -\frac14\psi_{U_c}$.
\end{Theorem}

\begin{remark}
 The proof uses only two facts about the operators $p^2$ and
$|x|^{-1}$:
First, $p^2$ is non-negative and, second, the localization
errors for $p^2$
fall-off faster than $|x|^{-1}$. The proof would work for many other
combinations having this essential relationship. For instance, as mentioned in
the introduction, $p^2$ could be replaced by $(p+A(x))^2$ for a suitable magnetic vector
potential $A(x)$.
\end{remark}

Theorem \ref{hohosi} follows by standard weak convergence
arguments from the following Proposition~\ref{heliumradius}.
The
details of the arguments, showing that the $H^1(\R^6)$-weak limit of
ground states $\psi_U$ for $U\to U_c$ is not zero, are worked out in
Section \ref{sec:exopt} on the PT equation, which is more complicated,
in fact.  Once one knows that the weak limit is not zero it is easy to
see that this weak limit has to be a ground state.

We know from \cite{Be} that $U_c>1$, and therefore for the
proof of Theorem \ref{hohosi} the assumption $U\geq
1+\delta$ in the following proposition is
not really a restriction. The main point is the uniformity
as $U\to U_c$.

\begin{proposition}[\textbf{Upper bound on the helium radius}]
\label{heliumradius}
For any $\delta>0$, there is a constant $C_\delta>0$ such that for all
$1+\delta\leq U < U_c$ and for all normalized ground states $\psi_U$ of
$H_{U}$ one has
\begin{equation}
 \label{eq:heliumradius}
\langle \psi_U |\, |x|_\infty^{-1}\, |  \psi_U \rangle \geq C_\delta \,,
\end{equation}
where $|x|_\infty = \max\{|x_1|,|x_2|\}$.
\end{proposition}

\begin{proof}
 Our goal will be to prove the
following operator inequality. (Actually, the analogous statement for
expectation values in the ground state would suffice for our purposes.) There is
a radius $\ell_0>0$ and constants $c,C>0$ such that for all larger radii
$\ell\geq\ell_0$ and for all $U\geq 1+\delta$ one has
\begin{equation}
 \label{eq:opineq}
H_U -E_U \geq - \frac{C}{\ell^2} \ \theta(\ell-|x|_\infty) 
+ \left(-\frac14-E_U + \frac{c}{|x|_\infty}\right)
\theta(|x|_\infty-\ell) \,.
\end{equation}
Here, $\theta$ is the Heaviside function, i.e., $\theta(t)=1$ if $t>0$ and
$\theta(t)=0$ if $t<0$.

Since $E_U\leq -\frac14$, this inequality, evaluated in a ground state $\psi_U$,
implies
$$
\frac{C}{\ell^2} \int_{\{|x|_\infty <\ell\}} \psi_U^2 \,dx \geq
c \int_{\{|x|_\infty >\ell\}} \frac{\psi_U^2}{|x|_\infty} \,dx 
\qquad
\text{for all}\
\ell\geq \ell_0 \,.
$$
Lemma \ref{calcu0}, stated below,
implies \eqref{eq:heliumradius}, and our proof is complete.

The crucial inequality \eqref{eq:opineq} has a simple physical interpretation.
The first term is a slightly negative energy, but vanishes relatively rapidly
like $\ell^{-2}$. This leaves us with the second term, which is like a potential
barrier. Although it goes to zero as $\ell\to \infty$, it does so much slower than $\ell^{-2}$ `on
average'.  The significance of this is quantified by the following lemma.

\begin{Lemma}[\textbf{Calculus lemma -- simple version}]\label{calcu0}
Let $\rho \in L^1(\R_+)$ be non-negative, with $\int_0^\infty \rho(r) dr = 1$.
Assume that there are constants $b>0$ and $\ell_c\geq 0$ such that
\begin{equation}\label{asu0}
\frac b{\ell^2} \int_0^\ell \rho(r) dr \geq
\int_\ell^\infty \frac 1 r \rho(r) dr 
\end{equation}
for all $\ell \geq \ell_c$. Then 
\begin{equation}\label{cts0}
\int_0^\infty \frac 1 r \rho(r) dr \geq \frac{1}{2(b+\ell_c)}
\end{equation}
and
\begin{equation}
 \label{eq:massinside}
\int_0^{\ell_c} \rho(r) dr \geq \frac{\ell_c^2}{(b+\ell_c)^2} \,.
\end{equation}
\end{Lemma}

A generalization of this lemma is given later in Lemma \ref{calcu}.

\begin{proof}[Proof of Lemma \ref{calcu0}]
We start by showing that
\begin{equation}
 \label{eq:calcbasic}
\int_{\ell_c}^\infty \frac 1 r \rho(r) dr \geq \frac{1}{b+\ell_c} -
\frac1{\ell_c} \int_0^{\ell_c} \rho(r) dr \,.
\end{equation}
To prove this, we define $f(R)=\int_0^R \rho(r)\,dr$ and rewrite the assumption
as $ b \ell^{-2} f(\ell) \geq \int_\ell^\infty r^{-1} \rho(r) dr$ for all
$\ell\geq\ell_c$. Therefore
\begin{align*}
 \int_{\ell_c}^\infty \frac 1 r \rho(r) dr 
& = \int_{\ell_c}^\infty \frac 1 r f'(r) dr 
= \int_{\ell_c}^\infty \frac 1{r^2} f(r) dr -\frac{f(\ell_c)}{\ell_c} \\
& \geq b^{-1} \int_{\ell_c}^\infty \int_r^\infty \frac 1 s \rho(s) ds \,dr
-\frac{f(\ell_c)}{\ell_c} \\
& = b^{-1} \int_{\ell_c}^\infty \frac {s-\ell_c} s \rho(s) ds
-\frac{f(\ell_c)}{\ell_c}\\
& = b^{-1} \left( 1-f(\ell_c) -\ell_c \int_{\ell_c}^\infty \frac
1 s \rho(s) ds \right) -\frac{f(\ell_c)}{\ell_c}\,.
\end{align*}
This is the same as \eqref{eq:calcbasic}.

In order to derive \eqref{cts0} from \eqref{eq:calcbasic} we distinguish two
cases according to whether $\ell_c^{-1} \int_0^{\ell_c} \rho(r) dr \leq
(2(b+\ell_c))^{-1}$ or not. In the first case,
\eqref{eq:calcbasic} yields
$$
\int_0^\infty \frac 1 r \rho(r) dr \geq \int_{\ell_c}^\infty \frac 1 r \rho(r)
dr 
\geq \frac{1}{2(b+\ell_c)} \,.
$$
In the opposite case, we have
$$
\int_0^\infty \frac 1 r \rho(r) dr \geq \int_0^{\ell_c} \frac 1 r \rho(r)
dr \geq \frac{1}{\ell_c} \int_0^{\ell_c} \rho(r)dr
\geq \frac{1}{2(b+\ell_c)} \,.
$$
This proves \eqref{cts0}.

Finally, to prove \eqref{eq:massinside} we combine the assumption
\eqref{asu0} at $\ell=\ell_c$ with \eqref{eq:calcbasic} to get
$$
\frac b{\ell_c^2} \int_0^{\ell_c} \rho(r) dr \geq
\int_{\ell_c}^\infty \frac 1 r \rho(r) dr
\geq \frac{1}{b+\ell_c} -
\frac1{\ell_c} \int_0^{\ell_c} \rho(r) dr \,,
$$
which is the same as \eqref{eq:massinside}.
\end{proof}

We now turn to the proof of the potential-theoretic
inequality \eqref{eq:opineq}.

\begin{Lemma}[\textbf{Partition of unity}]\label{pou}
 For any $0<\epsilon<1/2$ there is a $C_\epsilon>0$ such that the following
holds for any $\ell>0$. There is a quadratic partition of unity
$$
\chi_0(x)^2+\chi_1(x)^2+\chi_2(x)^2+\chi_3(x)^2 =1
\qquad\text{for all}\ x=(x_1,x_2)\in\R^3\times\R^3\,,
$$
such that
\begin{align*}
 & \chi_0(x)=0 \ \text{unless}\quad |x|_\infty\leq \ell\,,\\
 & \chi_1(x)=0 \ \text{unless}\quad |x|_\infty\geq \ell/2\,,\ |x|_\infty \leq
(1-\epsilon)|x_1-x_2|\,,\\
 & \chi_2(x)=0 \ \text{unless}\quad |x|_\infty\geq \ell/2\,,\ |x|_\infty \geq
(1-2\epsilon)|x_1-x_2|\,,\  |x_1|\leq (1+\epsilon)|x_2|\,,\\
 & \chi_3(x)=0 \ \text{unless}\quad |x|_\infty\geq \ell/2\,,\ |x|_\infty \geq
(1-2\epsilon)|x_1-x_2|\,,\  |x_2|\leq (1+\epsilon)|x_1|\,,
\end{align*}
and
\begin{align*}
 \sum_j |\nabla\chi_j|^2 & \leq \frac{C_\epsilon}{\ell^2}
\qquad\textrm{when}\quad \chi_0(x)^2>0 \,, \\
\sum_j |\nabla\chi_j|^2 & \leq \frac{C_\epsilon}{\ell\, |x|_\infty}
\qquad\text{when}\quad \chi_1(x)^2+\chi_2(x)^2+\chi_3(x)^2>0 \,.
\end{align*}
\end{Lemma}

One can choose $\chi_0$ and $\chi_1$ to be symmetric functions of $(x_1,x_2)$
and one can choose $\chi_2(x_1,x_2)=\chi_3(x_2,x_1)$. This preserves both Bose
and Fermi statistics.

\begin{proof}[Sketch of proof of Lemma \ref{pou}]
 The localization error $\sum_j |\nabla\chi_j|^2$ is supported in three regions,
 $\{\ell/2 \leq |x|_\infty\leq \ell\}$, $\{(1-2\epsilon)|x_1-x_2|\leq |x|_\infty
\leq (1-\epsilon)|x_1-x_2|\}$ and $\{ (1+\epsilon)^{-1} |x_2| \leq |x_1|\leq
(1+\epsilon)|x_2| \}$. It is a geometric question to check the stated
inequalities on the supports of the different $\chi$'s. We leave this to the
reader. The fact that $C_\epsilon$ is independent of $\ell$ is a consequence of
scaling.
\end{proof}

\begin{proof}[Proof of inequality \eqref{eq:opineq}]
 By Lemma \ref{pou} and the IMS localization formula we can write, for any wave
function $\Psi$,
\begin{equation*}
\langle \Psi | H_U |\Psi \rangle  
= \sum_{j=0}^3 \left\langle \Psi_j \left| H_U - \sum_{k=0}^3 \left|
\nabla \chi_k \right|^2 \right|  \Psi_j \right\rangle
 =: \sum_{j= 0}^3 e_j \|\Psi_j\|^2
\end{equation*}
with $\Psi_j(x_1,x_2) = \Psi(x_1,x_2) \chi_j(x_1,x_2)$ and with numbers
$e_j$ (depending on $\Psi_j$). 
In the following, we shall derive lower bounds on $e_j$. 

For $j=0$, we simply use the bound on the localization error from Lemma
\ref{pou} to bound
$$
e_0 \geq E_U - \frac{C_\epsilon}{\ell^2} \,.
$$
For $j=1$, we use the fact that on the support of $\chi_1$, if, say,
$|x_2|=|x|_\infty$ then
$$
|x_1| \geq \frac\epsilon{1-\epsilon} |x_2| = \frac\epsilon{1-\epsilon}
|x|_\infty \,.
$$
Therefore, for any $\lambda\geq0$,
$$
-|x_1|^{-1} - |x_2|^{-1} \geq -\epsilon^{-1} |x|_\infty^{-1} \geq -
2(\epsilon^{-1}+\lambda)\ell^{-1}+ \lambda |x|_\infty^{-1}
$$
and
$$
e_1 \geq - 2(\epsilon^{-1}+\lambda)\ell^{-1} + \|\Psi_1\|^{-2} \left\langle
\Psi_1 \left| \frac{\lambda -C_\epsilon \ell^{-1}}{|x|_\infty}
\right|\Psi_1\right\rangle \,.
$$
For $j=2$ we bound $p_1^2-|x_1|^{-1} \geq -\frac14$. Moreover, in that region
$$
|x_2| \geq (1+\epsilon)^{-1} |x|_\infty 
\qquad\text{and}\qquad
|x_1-x_2| \leq (1-2\epsilon)^{-1} |x|_\infty \,,
$$
so that
$$
e_2 \geq -\frac14 + \|\Psi_2\|^{-2} \left\langle \Psi_2 \left| 
\frac{U(1-2\epsilon)-1-\epsilon-C_\epsilon \ell^{-1}}{|x|_\infty}
\right|\Psi_2\right\rangle \,.
$$
Similarly, for $j=3$,
$$
e_3 \geq -\frac14 + \|\Psi_3\|^{-2} \left\langle \Psi_3 \left| 
\frac{U(1-2\epsilon)-1-\epsilon-C_\epsilon \ell^{-1}}{|x|_\infty}
\right|\Psi_3\right\rangle \,.
$$

In order to specify the free parameters $\epsilon$, $\ell$ and $\lambda$ we
restrict ourselves to $U\geq 1+\delta$ with some fixed $\delta>0$. (We will not
reflect the dependence on $\delta$ in our notation.) First, we choose $\epsilon$
so small that $U(1-2\epsilon)-1-\epsilon \geq \delta/2$ and we choose
$\lambda=\delta/2$. Then we choose $\ell_0$ so large that
$2(\epsilon^{-1}+\lambda)\ell_0^{-1}\leq 1/4$ and $C_\epsilon
\ell_0^{-1}\leq\delta/4$. With these choices, we have
$$
e_j \geq -\frac14 + \|\Psi_j\|^{-2} \left\langle \Psi_j \left| 
\frac{\delta/4}{|x|_\infty} \right|\Psi_j\right\rangle
$$
for any $\ell\geq\ell_0$ and any $j=1,2,3$. This, together with the fact that
$$
\chi_0(x)^2 \leq \theta(\ell-|x|_\infty)
$$
implies \eqref{eq:opineq}.
\end{proof}

This completes our proof of Proposition \ref{heliumradius}.
\end{proof}

\begin{remark}
 One can ask a similar question for $N$ electrons and whether there is binding
at the (unique) $U=U_c$ such that
\begin{equation}
 \label{eq:ucritatom}
E^{(N)}_{U_c} = E^{(N-1)}_{U_c} \,.
\end{equation}
(Here $E^{(n)}_U$ is the ground state energy of $n$ electrons with repulsion
strength $U$.) We denote by $M$ the smallest number $n\leq N-1$ such that
$E^{(N)}_{U_c} = E^{(n)}_{U_c}$ and we recall that by Zhislin's theorem
\cite{Zh} we have $U_c\geq M^{-1}$. Whether equality or strict inequality holds
has to be
decided by an independent variational calculation (as Bethe \cite{Be} did for
helium). We have nothing to say about the case $U_c= M^{-1}$. On the other hand,
if $U_c> M^{-1}$ our method should work and show that there
is a bound state for the $N$-electron system. Again, the
method should extend to magnetic fields, pseudo-relativistic
models with positive mass, spin-polarized systems, etc.
\end{remark}

Another area in which to try to utilize our method is Hartree or Hartree-Fock
theory, though we have not pursued this.


\section{The Bipolaron}

We now return to our main theme and consider the simplest interesting case,
namely two polarons, whose Hamiltonian is 
\begin{equation}
 \label{eq:ham}
H_U^{(2)} = p_1^2 + p_2^2 - \sqrt\alpha \phi(x_1) - \sqrt\alpha \phi(x_2)  +H_f + \frac{U}{|x_1-x_2|} \,.
\end{equation}
We have shown in \cite{FLST} that there is a critical $U_c(\alpha)
< \infty$ such that $E^{(2)}_U(\alpha) = 2 E^{(1)}(\alpha)$ for all
$U\geq U_c(\alpha)$, while $E_U^{(2)}(\alpha)< 2 E^{(1)}(\alpha)$ for
all $U<U_c(\alpha)$ (by concavity). It is easy to see that
$U_c(\alpha)\geq 2 \alpha$. In the following, we {\em assume} that
$U_c(\alpha) > 2 \alpha$. This is true for large $\alpha$, since
$\liminf_{\alpha\to \infty} U_c(\alpha)/\alpha \gtrsim 2.3$ due to the
convergence to the Pekar-Tomasevich functional in the strong coupling
limit \cite{DoVa,LiTh,spohn}, for which the critical ratio is known to be at least $\approx 2.3$ \cite{SuMo,fomin,VSPD}.

\begin{Theorem}[\textbf{Upper bound on the bipolaron radius}]\label{thm:binding}
For any $\epsilon>0$, there is a constant $C_\epsilon>0$ such that for all $0<
2\alpha(1+\epsilon) < U < U_c(\alpha)$ and for all states
$\Psi$  one has
\begin{equation}  \label{eq:bindabs}
\left\langle \Psi \left| \frac 1{|x_1-x_2|} \right| \Psi \right \rangle \geq
\frac { U - 2 \alpha(1+\epsilon)}{C_\epsilon ( 1 + U /\alpha)} \frac{
\left\langle \Psi \left| 2 E^{(1)}(\alpha) - H^{(2)}_U  \right| \Psi \right
\rangle}{2E^{(1)}(\alpha) - E_U^{(2)}(\alpha)} \,.
\end{equation}
\end{Theorem}

We emphasize two consequences of Theorem \ref{thm:binding}:

\begin{enumerate}
 \item In any approximate ground state, in the sense that 
\begin{equation}
 \left\langle \Psi \left| H^{(2)}_U  \right| \Psi \right\rangle \leq (1-\lambda)
2E^{(1)}(\alpha) + \lambda E_U^{(2)}(\alpha)
\end{equation}
for some $\lambda > 0$, the expectation value of $|x_1-x_2|^{-1}$ is
uniformly bounded from below by a positive number as $U$ increases to
$U_c(\alpha)$. In particular, \emph{the size of the bipolaron does not
increase indefinitely as $U \to U_c(\alpha)$}.
\item The bipolaron energy $E^{(2)}_U(\alpha)$ is
\emph{not differentiable in $U$ at $U = U_c(\alpha)$}. While the right
derivative is zero, the left derivative is at least as big as
$(U_c(\alpha)-2\alpha(1+\epsilon))/[C_\epsilon(1+ U_c(\alpha)/\alpha)]$, for any
$\epsilon > 0$.
\end{enumerate}

\begin{proof}
The proof of Theorem~\ref{thm:binding} will be divided into three parts.

\emph{Step 1. Partition of the interparticle distance.} As in
\cite{FLST} we choose a quadratic partition of unity and localize the particles according to their
relative distance.
In order to construct this partition, we pick some parameters $b>1$ and $\ell > 0$, and let
\begin{equation}\label{phi1}
\varphi(t) := \left\{ 
\begin{array}{ll} 
0 & \text{for $t\leq \ell/b$} \,, \\
\sin\frac \pi 2 \frac{t-\ell/b}{\ell-\ell/b} & \text{for $\ell/b\leq t \leq \ell $} \,,\\
\cos\frac \pi 2\frac{t-\ell}{b\ell-\ell} & \text{for $\ell\leq t\leq b\ell$} \,,\\
0 & \text{for $t\geq b\ell$} \,. 
\end{array} \right.
\end{equation}
For $j\geq 1$, let $\varphi_j(t) := \varphi(b^{1-j}t)$, and for $j=0$, let 
\begin{equation}\label{phi2}
\varphi_0(t) := \left\{ \begin{array}{ll} 1 & \text{for $t\leq \ell/b$} \,, \\ \cos\frac \pi 2 \frac{t-\ell/b}{\ell-\ell/b} & \text{for $\ell/b\leq t\leq \ell$} \,, \\ 0 & \text{for $t\geq \ell$}\,. \end{array}  \right.
\end{equation}
Then
\begin{equation}\label{phno}
\sum_{j\geq 0} \varphi_j(t)^2 = 1 \quad \text{for all $t\geq 0$} \,.
\end{equation}

Using the IMS localization formula, we can write, for any wave function $\Psi$,
\begin{equation}\label{2.4}
\langle \Psi | H^{(2)}_U |\Psi \rangle  
= \sum_{j\geq 0} \left\langle \Psi_j \left| H^{(2)}_U - 2 \sum_{k\geq 0} \left| \varphi_k'(|x_1-x_2|)\right|^2 \right|  \Psi_j \right\rangle
 =: \sum_{j\geq 0} e_j \|\Psi_j\|^2
\end{equation}
with $\Psi_j(x_1,x_2) = \Psi(x_1,x_2) \varphi_j(|x_1-x_2|)$ and with numbers $e_j$ (depending on $\Psi_j$). 
In the following, we shall derive lower bounds on $e_j$. 
For our bounds we shall use the fact that on the support of $\varphi_j(|x_1-x_2|)$, the localization error is dominated by
\begin{equation}\label{eq:locerror}
 \sum_{k\geq 0} \left| \varphi_k'(|x_1-x_2|)\right|^2 \leq \frac{\pi^2}{4(\ell-\ell/b)^2} \times \left\{ \begin{array}{ll} 1 & \text{if $j=0$} \,, \\ b^{2(1-j)} & \text{if $j\geq 1$} \,. \end{array}\right.
\end{equation}

\bigskip

\emph{Step 2. Estimate of localized energies.}
Because of \eqref{eq:locerror} we can simply bound 
\begin{equation}\label{eq:e0}
e_0 \geq E_U^{(2)}(\alpha) - \frac{\pi^2}{2(\ell-\ell/b)^2} \,.
\end{equation}
For $j\geq 1$, we further localize each of the two particles to its own ball of radius $b^j L$ for some parameter $L>0$. This will entail an additional localization error. Concretely, let 
\begin{equation}\label{defchi}
\chi(x) = \frac 1{\sqrt{2\pi} |x|}
\left\{ 
\begin{array}{ll} 
\sin(\pi|x|) & \text{for $|x|\leq 1$} \,, \\
0 & \text{for $|x|\geq 1$} \,,
\end{array}\right.
\end{equation}
and note that $\int dx \,\chi(x)^2 = 1$ and $\int dx |\nabla \chi(x)|^2 = \pi^2$. With
\begin{equation}
\Psi_{j,u_1,u_2}(x_1,x_2) = \Psi_j(x_1,x_2) (b^j L)^{-3} \chi(b^{-j}(x_1-u_1)/L) \chi(b^{-j}(x_2-u_2)/L)
\end{equation}
we have, by a continuous version of the IMS localization formula, 
\begin{equation}\label{48}
  \left\langle \Psi_j \left| H^{(2)}_U  \right|  \Psi_j \right\rangle  = \int_{\R^3} du_1 \int_{\R^3} du_2 \left\langle \Psi_{j,u_1,u_2} \left| H^{(2)}_U   -  \frac {2\pi^2}{b^{2j}L^2} \right|  \Psi_{j,u_1,u_2}\right\rangle \,. 
\end{equation}
Note that since $b^j\ell\geq |x_1-x_2|\geq b^{j-2} \ell$ on the support of $\varphi_j$, the wave function $\Psi_{j,u_1,u_2}$ is non-zero only if the distance $d$ between the two balls of radius $b^j L$ centered at $u_1$ and $u_2$, respectively, satisfies 
\begin{equation}\label{bounddd}
d \in D_j : = [ b^{j-2}\ell-4 b^jL,\ell b^j ] \,.
\end{equation}
We shall choose $L< \ell/(4b^2)$.

\begin{Lemma}\label{lem3}
Assume that $\Psi$ is normalized and supported in $B_1\times B_2$ where $B_1$ and $B_2$ are disjoint balls of some radius $R$, separated a distance $d$. Then
\begin{equation}\label{ls}
\langle \Psi| H^{(2)}_{U}|\Psi\rangle \geq 2 E^{(1)}(\alpha) - \frac {2\alpha} d + \frac{U}{d+4R}\,.
\end{equation}
\end{Lemma}

Eq.~\eqref{ls}  is an easy consequence of the functional integral representation of the ground state energy.
The proof can be found in \cite[Lemma~1]{FLST}.

We shall apply inequality \eqref{ls} to (\ref{48}), with $U$ replaced by $2\alpha(1+\epsilon)$ for some $\epsilon>0$. Using also (\ref{eq:locerror}) we conclude that 
\begin{align}\nonumber
e_j & \geq 2 E^{(1)}(\alpha) + \|\Psi_j\|^{-2} \left\langle \Psi_j \left|  \frac{U-2\alpha(1+ \epsilon)}{|x_1-x_2|} \right|\Psi_j\right\rangle \\ &  \quad + \min_{d\in D_j} \left( - \frac{2\alpha}{d} + \frac{2\alpha(1+\epsilon)}{ d + 4 b^j L} \right)  - b^{-2j} \left( \frac {b^2 \pi^2}{2(\ell-\ell/b)^2} + \frac {2\pi^2}{L^2}\right) \,.\label{tcL}
\end{align}
For given $\epsilon>0$ (and given $b>1$), we can choose $L$ to be a small enough
constant times $\ell$ such that, as long as $\ell$ is large enough, the sum of
the terms in the second line of (\ref{tcL}) is positive. More precisely, for
given $\epsilon>0$ (and given $b>1$) we can choose $\delta>0$ small enough such
that
\begin{equation}
\kappa_\epsilon := \min_{b^{-2}-4\delta \leq d\leq 1} \left( -\frac 1 d + \frac{1+\epsilon}{d+4 \delta} \right)>0\,.
\end{equation}
With $L = \delta \ell$, we see that 
\begin{equation}
\min_{d\in D_j} \left( - \frac{2\alpha}{d} + \frac{2\alpha(1+\epsilon)}{ d + 4 b^j L} \right)  - b^{-2j} \left( \frac {b^2 \pi^2}{2(\ell-\ell/b)^2} + \frac {2\pi^2}{L^2}\right) \geq 0
\end{equation}
for all $j\geq 1$ if 
\begin{equation}
\ell \geq \ell_c : = \frac{\pi^2}{\alpha \kappa_\epsilon} \left( \frac b{4 ( 1 - b^{-1})^2 } + \frac 1{b\delta^2} \right)\,.
\end{equation}
Under this condition, we thus have 
\begin{equation}\label{bej}
e_j  \geq 2 E^{(1)}(\alpha) + \|\Psi_j\|^{-2} \left\langle \Psi_j \left|  \frac{U-2\alpha(1+ \epsilon)}{|x_1-x_2|} \right|\Psi_j\right\rangle
\end{equation}
for all $j \geq 1$.

\bigskip

\emph{Step 3. Upper bound on the particle distance.} From \eqref{2.4} and our bounds (\ref{eq:e0}) and (\ref{bej}) on $e_j$ we conclude that, for all $\ell\geq \ell_c$, 
\begin{align}\nonumber
 H^{(2)}_U - 2E^{(1)}(\alpha)  & \geq \left(  E_U^{(2)}(\alpha) -2 E^{(1)}(\alpha) - \frac{\pi^2}{2 \ell^2 (1-b^{-1})^2} \right) \varphi_0(|x_1-x_2|)^2  \\ \nonumber & \quad +  \frac {U-2\alpha(1+\epsilon) }{|x_1-x_2|} \left( 1 - \varphi_0(|x_1-x_2|)^2 \right) \\ \nonumber & \geq \left(  E_U^{(2)}(\alpha) -2 E^{(1)}(\alpha) - \frac{\pi^2}{2 \ell^2 (1-b^{-1})^2}\right) \theta(\ell - |x_1-x_2|)   \\  & \quad +   \frac {U-2\alpha(1+\epsilon)}{|x_1-x_2|} \theta(|x_1-x_2| - \ell)\,.
\end{align}
Now let $\Psi$ satisfy
\begin{equation}\label{asps}
\langle\Psi| H^{(2)}_U - 2 E^{(1)}(\alpha) | \Psi\rangle \leq \lambda   \left(  E_U^{(2)}(\alpha) -2 E^{(1)}(\alpha) \right)
\end{equation}
for some $0<\lambda<1$. For such a $\Psi$, the previous inequality yields the bound
\begin{align}\nonumber
&\left\langle \Psi \left|  \left[ (1-\lambda)\left(2 E^{(1)}(\alpha) -  E_U^{(2)}(\alpha) \right) + \frac{\pi^2}{2 \ell^2 (1-b^{-1})^2}  \right] \theta(\ell - |x_1-x_2|) \right| \Psi \right\rangle \\ & \geq \left\langle \Psi \left| \left[ \lambda \left(2 E^{(1)}(\alpha) -  E_U^{(2)}(\alpha) \right) +   \frac {U-2\alpha(1+\epsilon)}{|x_1-x_2|}\right] \theta(|x_1-x_2| - \ell)\right| \Psi\right\rangle \,. \label{tat}
\end{align}
To conclude the proof, we need the following lemma.

\begin{Lemma}[\textbf{Calculus lemma -- general version}]\label{calcu}
Let $\rho \in L^1(\R_+)$ be non-negative, with $\int_0^\infty \rho(r) dr = 1$. Assume that there are constants $a\geq 0$, $b>0$, $c>0$, $0<\lambda\leq 1$ and $\ell_c\geq 0$ such that 
\begin{equation}\label{asu}
\left( (1-\lambda) a + \frac b{\ell^2}\right) \int_0^\ell \rho(r) dr \geq \int_\ell^\infty \left( \lambda a + \frac c r \right)\rho(r) dr 
\end{equation}
for all $\ell \geq \ell_c$. Then 
\begin{equation}\label{cts}
\int_0^\infty \frac 1 r \rho(r) dr \geq \frac {\lambda c}{b + 2 c \ell_c}\,.
\end{equation}
\end{Lemma}

\begin{proof}
Let $f(r) = \int_0^r \rho(s) ds$, and define $0<\ell_0\leq \infty$ by $f(\ell_0) = \lambda$. In case $\ell_0\leq \ell_c$, we have $f(\ell_c)\geq \lambda$ and hence
\begin{equation}\label{ua}
\int_0^\infty \frac 1 r \rho(r) dr \geq \frac{f(\ell_c)}{\ell_c} \geq \frac \lambda {\ell_c}\,.
\end{equation}
In particular, (\ref{cts}) holds. We can thus assume $\ell_0 > \ell_c$. For all $\ell\leq \ell_0$, we have
\begin{equation}
(1-\lambda) f(\ell) \leq \lambda (1-f(\ell)) 
\end{equation}
and hence (\ref{asu}) implies that 
\begin{equation}\label{cons}
 \frac b{\ell^2} f(\ell)  \geq c \int_\ell^\infty  \frac 1 r \rho(r) dr 
\end{equation}
for all $\ell_c\leq \ell \leq \ell_0$. Using integration by parts,  
\begin{equation}
\int_0^\infty \frac 1 r \rho(r) dr \geq \frac{f(\ell_c)}{\ell_c} + \int_{\ell_c}^{\ell_0} \frac 1 r f'(r) dr = \frac{f(\ell_0)}{\ell_0} + \int_{\ell_c}^{\ell_0} \frac 1 {r^2} f(r) dr \,.
\end{equation}
The first term on the right can be dropped for a lower bound, and in the integrand of the second we can use (\ref{cons}). We obtain
\begin{equation}
\int_0^\infty \frac 1 r \rho(r) dr \geq \frac c b \int_{\ell_c}^{\ell_0}  dr \int_r^\infty \frac{ f'(s)}{s} ds =  \frac c b \int_{\ell_c}^{\infty}  \frac{ f'(s)}{s}  \left( \min\{s,\ell_0\} -\ell_c \right) ds\,.
\end{equation}
For a lower bound, we can restrict the integral on the right to $s\leq \ell_0$.
Using (\ref{ua}) again, this yields
\begin{align}\nonumber
\int_0^\infty \frac 1 r \rho(r) dr & \geq   \frac c b \int_{\ell_c}^{\ell_0} 
\frac{ f'(s)}{s}  \left( s -\ell_c \right) ds \\ \nonumber & = \frac c b \left(
f(\ell_0) - f(\ell_c) \right) - \frac c b \ell_c \int_{\ell_c}^{\ell_0} \frac 1
s \rho(s) ds \\ & \geq \frac c b \left (\lambda - 2 \ell_c \int_0^\infty \frac
1r \rho(r) dr \right)\,,
\end{align}
which proves the lemma.
\end{proof}

We apply this lemma to (\ref{tat}), with $\rho$ being the probability
distribution of $|x_1-x_2|$ in the state $\Psi$. We conclude that, if $\Psi$
satisfies (\ref{asps}), then
\begin{equation}
\left\langle \Psi \left|  \frac 1{|x_1-x_2|}  \right| \Psi \right\rangle  \geq \frac {\lambda( U - 2 \alpha(1+ \epsilon))}{\frac{\pi^2}{2(1-b^{-1})^2} + 2 \ell_c  (U-2\alpha(1+\epsilon))}\,.
\end{equation}
This implies \eqref{eq:bindabs}.
\end{proof}


\section{Many Polarons}

We recall that the Hamiltonian $H^{(N)}_U$ for $N$ particles is given by
\eqref{eq:ham} and that its ground state and minimum break-up energy were
defined in \eqref{eq:gse} and \eqref{eq:mbe}. 
The analogue of Theorem~\ref{thm:binding} in the $N$-particle case is as
follows.

\begin{Theorem}[\textbf{Upper bound on the $N$-polaron
radius}]\label{thm:bindingN}
For any $N\geq 2$ and  $\epsilon>0$ there is a constant $C_\epsilon(N)>0$  such that for all
$0< 2\alpha(1+\epsilon) < U$ with $ E^{(N)}_U(\alpha) < \widetilde E_U^{(N)}(\alpha)$
and all states $\Psi$
\begin{equation} \label{eq:bindabs2} \left\langle \Psi \left| \frac
      1{\max_{i\neq j}|x_i-x_j|} \right| \Psi \right \rangle \geq
  \frac { U - 2 \alpha(1+\epsilon)}{C_\epsilon(N) ( 1 + U /\alpha)}
  \frac{ \left\langle \Psi \left| \widetilde E^{(N)}_U(\alpha) -
        H^{(N)}_U \right| \Psi \right \rangle}{\widetilde
    E^{(N)}_U(\alpha) - E_U^{(N)}(\alpha)} \,.
\end{equation}
\end{Theorem}

Since (\ref{eq:bindabs2}) holds for all $\Psi$, it can be
reformulated as an operator inequality. The bound is
non-trivial only for states
$\Psi$ with $\langle\Psi|H^{(N)}_U|\Psi\rangle < \widetilde
E^{(N)}_U(\alpha)$, however, which exist since $ E^{(N)}_U(\alpha) < \widetilde
E_U^{(N)}(\alpha)$ by assumption. For approximate ground states, satisfying
$$
\langle\Psi|H^{(N)}_U|\Psi\rangle \leq (1-\lambda) \widetilde
E^{(N)}_U(\alpha) + \lambda E^{(N)}_U(\alpha)
$$
for some $\lambda >0$, (\ref{eq:bindabs2}) gives a uniform upper bound
on the radius of the multipolaron system. The bound depends only on the value
of $\lambda$ and does not explode as $U\to U_c$.

\begin{proof}
  We perform a localization similar to that in the two-polaron case,
  but with $|x_1-x_2|$ replaced by the maximum of $|x_i - x_j|$ over
  all particle pairs $(i,j)$.  Let $\varphi_i$ be given as in
  (\ref{phi1})--(\ref{phi2}), for some $\ell>0$ and $b>1$.  We shall
  apply (\ref{phno}) with
\begin{equation}
t= \max_{i\neq j} |x_i-x_j|\,.
\end{equation}
  By the IMS
  localization formula,
\begin{align}\nonumber
\langle \Psi | H^{(N)}_U |\Psi \rangle  & = \sum_{j\geq 0} \left\langle \Psi \varphi_{j}(t) \left| H^{(N)}_U - \sum_{i=1}^N  \sum_{k\geq 0}   \left|\nabla_i  \varphi_{k}(t)\right|^2 \right|  \Psi \varphi_{j}(t)  \right\rangle \\ & =: \sum_{j\geq 0} e_j \| \Psi \varphi_j(t) \|^2\,.
\end{align}
Moreover, for almost every $X\in \R^{3N}$, 
\begin{equation}
\sum_{i=1}^N  \sum_{k\geq 0}   \left|\nabla_i  \varphi_{k}(t)\right|^2 = 2 \sum_{k\geq 0}   |\varphi_{k}'(t)|^2\,,
\end{equation}
which can be bounded as in (\ref{eq:locerror}) on the support of $\varphi_j$. In particular,
\begin{equation}
e_0 \geq E^{(N)}_U(\alpha)  - \frac{\pi^2}{2(\ell-\ell/b)^2}\,.
\end{equation}

For $j\geq 1$, we now proceed with the one-particle localization as in the
two-polaron case, localizing particle $i$ in a ball of radius $b^{j}
L$ centered at $u_i$, for suitably chosen $L > 0$. 
More precisely, with $\chi$ given in (\ref{defchi}), let
\begin{equation}
\Psi_{j,\bfu}(X) = \Psi(X) \varphi_j (t) (b^{j}L)^{-3N/2} \prod_{i=1}^N \chi(b^{-j}(x_i-u_i)/L)
\end{equation}
where we denote $\bfu=(u_1,\dots,u_N)$. We have
\begin{equation}\label{norm}
\|\Psi\varphi_j(t)\|^2 =  \int_{\R^{3N}} d\bfu\, \|\Psi_{j,\bfu}\|^2
\end{equation}
and, again using the IMS formula,
\begin{equation}
\langle \Psi \varphi_j(t) |H^{(N)}_U|\Psi \varphi_j(t) \rangle =  \int_{\R^{3N}} d\bfu  \left\langle \Psi_{j,\bfu} \left| H^{(N)}_U - \frac{N \pi^2}{b^{2 j}L^2}   \right|  \Psi_{j,\bfu}  \right\rangle  \,. \label{fl}
\end{equation}

All particles are now supported in balls $B_i$ of radius $b^{j} L$
centered at $u_i$.  Moreover, $\Psi_{j,\bfu}$ is nonzero only if
$\max_{i\neq k} |u_i - u_k| \geq b^{j-2}\ell - 2 b^j L$. In
particular, if we draw open balls of radius $R$ around all the $u_i$,
the resulting set is not connected as long as $(N-1) 2 R\leq
b^{j-2}\ell - 2 b^j L$.  Hence it is possible to split the particles
into two clusters, in such a way that the distance between any two
balls in different clusters is at least as big as
\begin{equation}\label{dijl}
\frac 1{N-1} \left( b^{j-2} \ell -2 b^j L\right) - 2 b^j L = \frac 1{N-1} \left( b^{j-2} \ell - 2 N b^j L\right)\,.
\end{equation}
We shall choose $L<\ell/(2Nb^2)$ to make this positive.

The analogue of Lemma~\ref{lem3} that we need now is the following.

\begin{Lemma}\label{lem3n}
Assume that $\Psi$ is normalized and supported in the set $B_1\times \cdots \times B_N$, where the $B_i$ are balls of radius $R$. Assume that, for some $1\leq n\leq N-1$, the distances $d_{ik}$ between balls $B_i$ for $i\leq n$ and $B_k$ for $k\geq n+1$ are positive. Then, for all $\epsilon>0$, 
\begin{align}\nonumber
\langle \Psi| H^{(N)}_{U}|\Psi\rangle & \geq  E^{(n)}_U(\alpha) + E^{(N-n)}_U(\alpha) \\ & \quad + \sum_{i=1}^n \sum_{k=n+1}^N  \left( \left\langle \Psi\left| \frac{ U -2\alpha(1+ \epsilon)}{|x_i-x_k|} \right| \Psi \right\rangle  - \frac {2\alpha} {d_{ik}} + \frac{2\alpha(1 +\epsilon)}{d_{ik}+4R}  \right) \,.
\end{align}
\end{Lemma}

Using the path integral representation of the ground state energy, the
proof is easy; we refer to \cite{FLST} for details.

We relabel the particles such that particles $1\leq i\leq n$  belong to one cluster, and $n+1\leq k\leq N$ to the other. The distances $d_{ik}$  are all bounded from below by (\ref{dijl}). They are also
bounded above by $\ell b^j$ on the support of $\Psi_j$. With 
\begin{equation}
D_j := \left[  (N-1)^{-1} \left( b^{j-2} \ell - 2 N b^j L\right)\,,\,   \ell b^j\right] 
\end{equation}
we conclude that 
\begin{align}\nonumber
e_j & \geq \widetilde E^{(N)}_U(\alpha) +(N-1) \|\Psi_j\|^{-2}  \left\langle \Psi_{j}\left| 
  \frac{ U -2\alpha(1+ \epsilon)}{\max_{i\neq k}|x_i-x_k|}  \right| \Psi_{j} \right\rangle \\ &  \quad + (N-1) \min_{d\in D_j} \left( - \frac{2\alpha}{d} + \frac{2\alpha(1+\epsilon)}{ d + 4 b^j L} \right)  - b^{-2j} \left( \frac {b^2 \pi^2}{2(\ell-\ell/b)^2} + \frac {N\pi^2}{L^2}\right) \,.
\end{align}
Here, we have assumed that the minimum over $d$ in the second line is non-negative, in which case we can use that $n(N-n) \geq (N-1)$ for $1\leq n \leq N-1$. We can now argue as in the two-polaron case that, for suitably chosen $L$, the sum of the terms in the second line is  positive for $\ell$ large enough. In fact, with $L= \delta \ell/(N-1)$ and 
\begin{equation}
\kappa_\epsilon := \min_{b^{-2}-2\delta N/(N-1) \leq d\leq N-1} \left( -\frac 1 d + \frac{1+\epsilon}{d+4 \delta} \right)
\end{equation}
(which is positive for small enough $\delta$) this is the case for all $j\geq 1$ if 
\begin{equation}
\ell \geq \ell_c : = \frac{\pi^2}{\alpha \kappa_\epsilon(N-1)^2} \left( \frac b{4 ( 1 - b^{-1})^2 } + \frac N{2b\delta^2} \right)\,.
\end{equation}
 
Recall that $t=\max_{i\neq k} |x_i-x_k|$.
We have thus shown the operator inequality 
\begin{align}\nonumber 
H_U^{(N)} - \widetilde E_U^{(N)}(\alpha) & \geq \left( E^{(N)}_U - \widetilde E_U^{(N)}(\alpha)- \frac{\pi^2}{2(\ell-\ell/b)^2} \right) \theta(\ell -t) \\ & \quad +(N-1) \frac{ U - 2\alpha(1+\epsilon)}{t} \theta (t -\ell)
\end{align}
for all $\ell\geq \ell_c$.  The remainder of the proof is as in the
two-particle case, applying Lemma~\ref{calcu} to the probability
distribution of $t$ in a state $\Psi$.
\end{proof}


\section{The Pekar-Tomasevich approximation}\label{sec:exopt}

Let us recall, for the reader's convenience, the definition and basic
properties of the PT model \cite{FLST}. There is a non-linear
differential-integral variational principle associated
with the polaron problem,  which gives the exact ground state energy in the
limit $\alpha\to \infty$. This variational problem was investigated in detail by
Pekar \cite{Pe}. 
Pekar and Tomasevich (PT) \cite{PeTo} generalized it to the bipolaron,  and the
extension to $N$-polarons obviously follows from \cite{PeTo}.

The PT functional is the result of a variational calculation and therefore gives
an upper bound to the ground state energy $E^{(N)}_U(\alpha)$. In order to
compute $\langle\Psi,H^{(N)}_U\Psi\rangle$, one takes a $\Psi$ of the form
$\psi\otimes \Phi$ where $\psi\in L^2(\R^{3N})$, $\Phi\in\mathcal F$,  and both
$\psi$ and $\Phi$ are normalized. For a given $\psi$ it is easy to compute the
optimum $\Phi$, and one ends up with the functional
\begin{align}\nonumber
\mathcal P^{(N)}_{U}[\psi] & := \sum_{i=1}^N \int_{\R^{3N}} |\nabla_i \psi|^2
\,dX + U \sum_{i<j} \int_{\R^{3N}} \frac{|\psi(X)|^2}{|x_i-x_j|} \,dX \\ & \quad
- \alpha \iint_{\R^3\times\R^3} \frac{\rho_\psi(x)\, \rho_\psi(y)}{|x-y|}
\,dx\,dy \,, \label{eq:pekar}
\end{align}
where $dX=\prod_{k=1}^N dx_k$, and
\begin{equation}
\rho_\psi(x) = \sum_{i=1}^N \int_{\R^{3(N-1)}} |\psi(x_1,\ldots,x,\ldots,x_N)|^2
\,dx_1\cdots \widehat{dx_i} \cdots dx_N
\end{equation}
with $x$ at the $i$-th position, and $\widehat{dx_i}$ meaning that $dx_i$ has to
be omitted in the product $\prod_{k=1}^N dx_k$. The ground state energy is
\begin{equation}
\mathcal E^{(N)}_U(\alpha) = \inf\left\{ \mathcal P^{(N)}_{U}[\psi] :
\int_{\R^{3N}} |\psi|^2\,dX = 1 \right\} \,.
\end{equation}
The variational argument above gives the upper bound
\begin{equation}
 E^{(N)}_U(\alpha) \leq \mathcal E^{(N)}_U(\alpha) = \mathcal
E^{(N)}_{U/\alpha}(1) \ \alpha^2 \,.
\end{equation}
(The equality follows by scaling.)
For $N=1$ this upper bound is due to Pekar; numerically, one has $\mathcal
E^{(1)}(\alpha) \approx -(0.109) \alpha^2$ \cite{Mi}. Moreover, the minimization
problem for $\mathcal E^{(1)}(\alpha)$ has a unique minimizer (up to
translations), see \cite{Li}.

The upper bound for $N=1$ was widely understood to be asymptotically
exact for large $\alpha$. A proof of this was finally achieved by
Donsker and Varadhan \cite{DoVa}, using the functional integral
representation and large deviation theory. Later, this fact
was rederived in
\cite{LiTh} by operator methods, and it was shown that the
error was
no worse than $\alpha^{9/5}$ for large $\alpha$. The fact that for
fixed ratio $\nu=U/\alpha\geq 0$ and for $N=2$ one has
\begin{equation}
 \label{eq:pekarlimit}
\lim_{\alpha\to\infty} \alpha^{-2} E^{(N)}_{U}(\alpha) = \mathcal
E^{(N)}_{\nu}(\alpha=1)
\end{equation}
was first noted in \cite{spohn}. This is also valid for arbitrary $N$.

We now address the question of whether the infimum $\mathcal
E^{(N)}_{U}(\alpha)$ is attained, that is, whether there is a
minimizer. It was proved in \cite{lewin} that this is the
case provided $\mathcal E^{(N)}_{U}(\alpha)<\tilde{\mathcal
  E}^{(N)}_{U}(\alpha)$. This minimum break-up energy is
defined as
before by
$$
\tilde{\mathcal E}^{(N)}_{U}(\alpha) = \min_{1\leq n\leq N-1}
\left( \mathcal E^{(n)}_U(\alpha) + \mathcal E^{(N-n)}_U(\alpha) \right)\,.
$$
Our next theorem gives an answer in the case of equality, i.e., when $\mathcal
E^{(N)}_{U}(\alpha)=\tilde{\mathcal E}^{(N)}_{U}(\alpha)$.

\begin{Theorem}\label{exopt}
 Let $U_c>0$ be such that $\mathcal E^{(N)}_{U_c}(\alpha)=\tilde{\mathcal
E}^{(N)}_{U_c}(\alpha)$ and assume that there is a sequence $U_n\to U_c$ such
that the infimum $\mathcal E^{(N)}_{U_n}(\alpha)$ is attained. Then the infimum
$\mathcal E^{(N)}_{U_c}(\alpha)$ is attained.
\end{Theorem}

It is important in our proof that any $U_c$ as in Theorem \ref{exopt} satisfies
$U_c>2\alpha$, as shown in \cite{lewin}. The key input in our proof is once
again a bound on the maximal distance between the particles.

\begin{proposition}\label{key}
For any $\nu>1$ and $N\geq 2$ there is an $\ell_0<\infty$ and a $\delta_0>0$
such that, whenever $U>2\alpha\nu$ and $\psi_U$ is a minimizer for $\mathcal
E^{(N)}_{U}(\alpha)$, then one has
\begin{equation}
 \label{eq:key}
\int_{\R^{3N}} \theta\big(\mbox{$\sum_{i,j}$} |x_i-x_j|^2 \leq \ell_0^2\big)\ \psi_U^2 \,dX \geq
\delta_0 \,.
\end{equation}
\end{proposition}

The proof is similar to that of Theorem \ref{thm:bindingN} and is omitted. We
only note that, since we assume that there is a minimizer, we can take
$\lambda=1$ in \eqref{asps} and therefore the simple version
\eqref{eq:massinside} of the calculus lemma suffices and yields \eqref{eq:key}.
One also needs to linearize the functional, as we did in \cite{FLST} and as we
will do later in the proof of Theorem~\ref{exopt}.

\begin{proposition}\label{nonzero}
 Let $U_n$ be as in Theorem \ref{exopt} and let $\psi_n$ be the corresponding
minimizers. Then there is a \emph{non-zero} $\psi\in H^1(\R^{3N})$ and a
sequence $\{a_n\}\subset\R^3$ such that a subsequence of
$\psi_n(x_1-a_n,\ldots,x_N-a_n)$ converges weakly in
$H^1(\R^{3N})$ to $\psi$.
\end{proposition}

\begin{proof}
 Let $\eta:[0,\infty)\to[0,1]$ be a smooth, compactly supported function with
$\eta(t)=1$ for $t$ in a neighborhood of zero. We shall prove the following:
There is an $R>0$, a non-zero $u\in H^1(\R^{3N})$ and a
sequence $\{a_n\}\subset\R^3$ such that a subsequence of
$u_n(x_1,\ldots,x_N) = \eta(\sum_{i,j}|x_i-x_j|^2/R^2)
\psi_n(x_1-a_n,\ldots,x_N-a_n)$ converges weakly in
$H^1(\R^{3N})$ to $u$. From
this, one easily derives the statement of the lemma.

We first note that $\|u_n\|\leq \|\psi_n\|= 1$ and $\sum_i \|\nabla_i u_n
\|^2 \leq \sum_i \|\nabla_i \psi_n \|^2 + C R^{-2}$ for some $C$ independent of
$n$. Hence the sequence $u_n$ is bounded in $H^1(\R^{3N})$, and the assertion
will follow by an easy extension of \cite[Thm. 8.10]{LiLo}, provided we can show
that $u_n$ does not converge to zero in measure.

To prove this, we use the Euler-Lagrange equation satisfied by the functions
$\psi_n$,
$$
\left( \sum_{n=1}^N \left( p_i^2 - \sqrt{\alpha}\, \phi_n(x_i) \right) + U
\sum_{i<j}
\frac1{|x_i-x_j|} \right) \psi_n = \lambda_n \psi_n \,,
$$
with $\phi_n= 2\sqrt\alpha\, \rho_{\psi_n}*|x|^{-1}$ and $\lambda_n = \mathcal
E_{U_n}^{(N)}(\alpha)-\alpha \iint \rho_{\psi_n}(x) |x-y|^{-1}
\rho_{\psi_n}(y)\,dx\,dy$. Multiplying the equation by
$\eta(\sum_{i, j}|x_i-x_j|^2/R^2)^2 \psi_n(x_1,\ldots,x_N)$ and integrating we
find that
\begin{align}\label{eq:energyest}
& \sum_{i=1}^N \int_{\R^{3N}} \left( |\nabla_i u_n|^2 -\sqrt{\alpha}\,
\phi_n(x_i)
u_n^2 \right) \,dX + U \sum_{i<j} \int_{\R^{3N}} \frac{u_n^2}{|x_i-x_j|}
\,dX \notag \\ 
& \qquad \leq (\lambda_n + C R^{-2} ) \int_{\R^{3N}} u_n^2 \,dX \,.
\end{align}
Since $\liminf_{n\to\infty} \lambda_n \leq \mathcal E_{U_c}^{(N)}(\alpha) <0$
we can choose $R>0$ sufficiently large such that $\liminf_{n\to\infty}
(\lambda_n+C R^{-2}) <0$. Moreover, by our key estimate \eqref{eq:key} we know
that, after increasing $R$ if necessary, we have
$$
\liminf_{n\to\infty} \int_{\R^{3N}} u_n^2 \,dX > 0 \,.
$$
(In order to apply Proposition \ref{key} we use the fact that $U_c>2\alpha$,
see \cite{lewin}.) From this we conclude that the only negative term on the left
side of \eqref{eq:energyest} cannot vanish in the limit, i.e.,
\begin{equation}
 \label{eq:lowerbd}
\liminf_{n\to\infty} \sum_{i=1}^N \int_{\R^{3N}} \phi_n(x_i) u_n^2 \,dX >0 \,.
\end{equation}

We now use the `pqr theorem' \cite[Ex. 2.22]{LiLo} to conclude from
\eqref{eq:lowerbd} that $\{u_n\}$ does not converge to zero in measure.
This theorem states, quite generally, that for any $1\leq p<q<r\leq\infty$ and
for any constant $C<\infty$ there are constants $\epsilon,\delta>0$ such that if
a function $f$ satisfies $\|f\|_p \leq C$, $\|f\|_r \leq C$ and $\|f\|_q \geq
C^{-1}$, then $|\{ |f| \geq \epsilon \}|>\delta$.

Returning to our concrete
situation, we note that, since $\{u_n\}$ is uniformly bounded in $H^1(\R^{3N})$, we
known from Sobolev inequalities that it is so in $L^p(\R^{3N})$ as well for any
$2\leq p\leq 6N/(3N-2)$.
To apply the `pqr theorem' we need to show that $\|u_n\|_p$ is
uniformly bounded \emph{away from zero} for some $2< p< 6N/(3N-2)$. First, note
that, since $(\sqrt{\rho_n})$ is uniformly bounded in $H^1(\R^3)$, $\phi_n$ is
uniformly bounded in $L^q(\R^3)$ for any $3<q\leq \infty$. Now we choose $C>0$
such that $\eta(t)=0$ for $t\geq C$ and we put
$$
g_n(X)=\theta(CR^{2}-\sum_{i,j}|x_i-x_j|^2 ) \sum_i \phi_n(x_i) \,.
$$
The sequence $\{g_n\}$ is uniformly bounded in $L^q(\R^{3N})$ for any $3<q\leq
\infty$, and by \eqref{eq:lowerbd} we have
$$
0<\delta \leq \sum_{i=1}^N \int_{\R^{3N}} \phi_n(x_i) u_n^2 \,dX 
= \int g_n u_n^2\,dX \leq \|g_n\|_q \|u_n^2\|_{q'}
$$
where $q^{-1}+q'^{-1}=1$. This shows that $(u_n)$ is uniformly bounded away
from zero in $L^{2q'}(\R^{3N})$, and we can choose $3<q<\infty$ in such a way
that
$2<2q'<6N/(3N-2)$. Hence the `pqr theorem' implies that $u_n$ does not tend to
zero in measure, which completes the proof of the lemma.
\end{proof}

We now turn to the proof of Theorem \ref{exopt}. An important ingredient, which
we have already used in \cite{FLST}, is that the energy functional can be
linearized. In order to state this precisely we define for any pair
$(\psi,\phi)\in H^1(\R^{3N})\cap\dot H^1(\R^3)$ (with $\psi$ not necessarily
normalized) the energy functional
\begin{align}\nonumber
\mathcal P_U[\psi,\phi] & := \sum_{i=1}^N \int_{\R^{3N}} \left( |\nabla_i
\psi|^2 - \sqrt{\alpha} \phi(x_i) |\psi|^2 \right)\,dX + U \sum_{i<j}
\int_{\R^{3N}} \frac{|\psi(X)|^2}{|x_i-x_j|} \,dX
\\ & \quad
+ \frac{1}{16\pi} \int_{\R^3} |\nabla \phi|^2 \,dx \int_{\R^{3N}} |\psi(X)|^2
\,dX  \,, \label{eq:pekarpot}
\end{align}
The crucial observation is that for any $\psi$ and $\phi$ one has 
\begin{align}\label{eq:linearize}
\mathcal P_U[\psi,\phi] 
\geq \|\psi\|^{2}\ \mathcal P^{(N)}_{U}[\|\psi\|^{-1}\psi] \,,
\end{align}
with equality if and only if $\phi=2\sqrt{\alpha} \
|x|^{-1}*\rho_{\psi/\|\psi\|}$. We are
now ready to give the

\begin{proof}[Proof of Theorem \ref{exopt}]
 Because of Proposition \ref{nonzero}, after passing to a subsequence and a
translation we may assume that $\{\psi_n\}$ converges weakly in
$H^1(\R^{3N})$ to $\psi\not\equiv 0$. Denoting $\tilde\psi_n = \psi_n - \psi$
and $\phi_n = 2\sqrt{\alpha}\ |x|^{-1}*\rho_{\psi_n}$ we claim that
\begin{equation}
 \label{eq:decomp}
\mathcal P_{U_n}[\psi_n,\phi_n] = \mathcal P_{U_c}[\psi,\phi_n] + \mathcal
P_{U_c}[\tilde\psi_n,\phi_n] + o(1) \,.
\end{equation}
Assuming \eqref{eq:decomp} for the moment, we now finish the proof of
Theorem \ref{exopt}. Indeed, in view of \eqref{eq:linearize}, we conclude from
\eqref{eq:decomp} that
$$
\mathcal E^{(N)}_{U_c}(\alpha) + o(1) 
\geq \|\psi\|^{2}\ \mathcal P^{(N)}_{U_c}[\|\psi\|^{-1}\psi]
+ \|\tilde\psi_n\|^{2}\ \mathcal
P^{(N)}_{U_c}[\|\tilde\psi_n\|^{-1}\tilde\psi_n] +o(1) \,.
$$
Since $\mathcal P^{(N)}_{U_c}[\|\tilde\psi_n\|^{-1}\tilde\psi_n]\geq \mathcal E^{(N)}_{U_c}(\alpha)$
and since $\|\tilde\psi_n\|^{2}=1-\|\psi\|^{2}+o(1)$, we learn
that
$$
\mathcal E^{(N)}_{U_c}(\alpha) +
o(1) \geq \mathcal P^{(N)}_{U_c}[\|\psi\|^{-1}\psi] +o(1)  \,,
$$
that is, $\|\psi\|^{-1} \psi$ is an optimizer at $U=U_c$.

It remains to prove \eqref{eq:decomp}. The analogous assertion separately for the terms
$\sum_{i=1}^N \int |\nabla_i \psi|^2 \,dX$ and $\sum_{i<j} \int
\frac{|\psi(X)|^2}{|x_i-x_j|} \,dX$ is an easy consequence of weak convergence.
Moreover,
$$
\int_{\R^3} |\nabla \phi_n|^2 \,dx
= \int_{\R^3} |\nabla \phi_n|^2 \,dx \int_{\R^{3N}} |\psi(X)|^2
+ \int_{\R^3} |\nabla \phi_n|^2 \,dx \int_{\R^{3N}} |\tilde\psi_n(X)|^2 + o(1)
$$
by the weak convergence of $\psi_n$ and the fact that $\int |\nabla \phi_n|^2
\,dx$ is uniformly bounded. Thus it remains to prove that
\begin{equation}
 \label{eq:decompzero}
\int_{\R^{3N}} \phi_n(x_i) \tilde\psi_n \psi \,dX = o(1)
\end{equation}
for each $i$. To prove this, we let $R>0$ and split the integral on the left
side according to whether $|X|_\infty=\max\{|x_1|,\ldots,|x_N|\}$ is bigger or smaller than $R$. In this
way we find that 
\begin{align}\nonumber
&\left| \int_{\R^{3N}} \phi_n(x_i) \tilde\psi_n \psi \,dX \right| 
\\ &\leq \sup_k \|\phi_k \|_\infty \left( \sup_k \|\tilde\psi_k\| \|\psi \ \theta(|X|_\infty \geq R) \| +
\|\tilde\psi_n \ \theta(|X|_\infty \leq R) \| \right) \,. \label{eq:apriori}
\end{align}
Here we also used that $\|\psi\|\leq 1$. The fact
that $\{\phi_n\}$ is uniformly bounded in $L^\infty(\R^3)$ follows from the fact that 
$\{\sqrt{\rho_{\psi_n}}\}$ is uniformly bounded in $H^1(\R^3)$. We can make the
first
term on the right side of \eqref{eq:apriori} as small as we like by choosing
$R$ large. On the other hand, for any fixed $R$, the sequence $\{\tilde \psi_n\,
\theta(|X|_\infty \leq R)\}$ tends to zero in $L^2(\R^{3N})$ by
the Rellich-Kondrashov Theorem \cite[Thm.~8.9]{LiLo}, and therefore we can make the second term on the right side
of \eqref{eq:apriori} as small as we like by choosing $n$ large. This completes
the proof of \eqref{eq:decompzero} and therefore of Theorem~\ref{exopt}.
\end{proof}

\bigskip {\it Acknowledgments.} We are grateful to Herbert Spohn for
making us aware of this problem. Partial financial support from the
U.S.~National Science Foundation through grant PHY-0965859 (E.L.)  and
the NSERC (R.S.) is gratefully acknowledged.


\end{document}